\documentclass[12pt]{article} 
\usepackage[fleqn]{amsmath}
\usepackage{amssymb}
\usepackage{color}
\usepackage{epsfig}
\usepackage{epsf}
\usepackage{ifthen}

\newboolean{Longversion}
\setboolean{Longversion}{false}

\newcommand*\widebar[1]{%
  \hbox{%
    \vbox{%
      \hrule height 0.5pt  
      \kern0.25ex
      \hbox{%
        \kern-0.1em
        \ensuremath{#1}%
        \kern-0.1em
      }%
    }%
  }%
} 

\newcommand*{\QEDA}{\hfill\ensuremath{\blacksquare}}
\newcommand{\set}[1]{\left\{ #1 \right\}}
\newcommand{\abs}[1]{\left\vert #1 \right\vert}
\newcommand{\Prob}[1]{\Pr\left( #1 \right)}
\newcommand{\Ex}[2][]{\mathbb{E}_{#1\!}\sqparen{#2}}
\newcommand{\Var}[1]{\operatorname{Var}\sqparen{#1}}
\newcommand{\Cov}[2]{\operatorname{Cov}\sqparen{#1, #2}}
\newcommand{\floor}[1]{\left\lfloor #1 \right\rfloor}
\newcommand{\paren}[1]{\left( #1 \right)}
\newcommand{\sqparen}[1]{\left[ #1 \right]}

\newcommand{\Unif}[2]{\operatorname{U}\!\paren{#1, #2}}
\newcommand{\Normal}[3][]{
\ifx\hfuzz#1\hfuzz 
\mathcal{N}\paren{#2, #3}
\else
\mathcal{N}\paren{#1; #2, #3}
\fi
}
\newcommand{\Beta}[2]{\operatorname{Beta}\paren{#1, #2}}

\newtheorem{theorem}{Theorem}
\newtheorem{corollary}[theorem]{Corollary}
\newtheorem{lemma}{Lemma}

\newtheorem{algorithm2}{Algorithm}

\newcounter{problemC}[section]

\begin{document}

\title{A Minimal Variance Estimator for the Cardinality of Big Data Set Intersection}
\author{Reuven Cohen~~~Liran Katzir~~~Aviv Yehezkel\\
Department of Computer Science\\
Technion\\
Haifa 32000, Israel\\
}
\date{}
\maketitle


\begin{abstract}
In recent years there has been a growing interest in developing ``streaming algorithms'' for efficient processing and querying of continuous data streams. These algorithms seek to provide accurate results while minimizing the required storage and the processing time, at the price of a small inaccuracy in their output. A fundamental query of interest is the intersection size of two big data streams. This problem arises in many different application areas, such as network monitoring, database systems, data integration and information retrieval. In this paper we develop a new algorithm for this problem, based on the Maximum Likelihood (ML) method. We show that this algorithm outperforms all known schemes and that it asymptotically achieves the optimal variance.
\end{abstract}


\section{Introduction} \label{sec:intro}

Classical processing algorithms for database management systems usually require several passes over (static) data sets in order to produce an accurate answer to a user query. However, for a wide range of application domains, the data set is very large and is updated on a continuous basis, making this approach impractical. 
For this reason, there is a growing interest in developing ``streaming algorithms'' for efficient processing and querying of continuous data streams in data stream management systems (DSMSs). These algorithms seek to provide accurate results while minimizing both the required storage and the processing time per stream element, at the price of a small inaccuracy in their output \cite{Beyer07, Cosma:2011, Dasu02, Metwally:2008}. Streaming algorithms for DSMSs typically summarize the data stream using a small sketch, and use probabilistic techniques in order to provide approximate answers to user queries.
Such big data streams appear in a wide variety of computer science applications. They are common, for example, in computer networks, where detailed usage statistics (such as the source IP addresses of packets) from different parts of the network need to be continuously collected and analyzed for various security and management tasks.

A fundamental query of interest is the intersection size of two big data streams.
Consider two streams of elements, $\widebar{A}$ and $\widebar{B}$, taken from two sets $A$ and $B$ respectively. Suppose that each element may appear more than once in each stream. Let $n=\abs{\widebar{A} \cap \widebar{B}}$. For example, for $\widebar{A}=a,b,c,d,a,b$ and $\widebar{B}=a,a,c,c$, we get $n=2$.
Finding $n$ is a problem that arises in many different application areas such as network monitoring, database systems, data integration and information retrieval \cite{Broder00, Broder97}. 

As an application example, $a_i \in A$ and $b_j \in B$ could be streams of IP packets passing through two routers, $R_1$ and $R_2$. In this case, $\abs{A \cap B}$ represents the number of flows forwarded by both routers. Thus, $\widehat{\abs{A \cap B}}$ can be used for security and traffic monitoring, e.g., on-line detection of denial of service attacks.

As another example, $a_i \in A$ and $b_j \in B$ could be sets of sub-strings found within two text documents. In this case, $\abs{A \cap B}$ represents the number of sub-strings shared by both documents, which can be viewed as the similarity of the documents. Among many other applications, this allows plagiarism detection and pruning of near-duplicate search results in search engines.

One can find the exact value of $n$ by computing the intersection set $C$ in the following way. For every element $b_i \in B$, compare $b_i$ to every $a_j \in A$ and $c_k \in C$. If $b_i \notin C$ and $b_i \in A$, add $b_i$ to $C$. After all the elements are treated, return the number of elements in $C$. This naive approach does not scale if storage is limited or if the sets are very large.
In these cases, the following estimation problem should be solved.
Given two streams of elements (with repetitions) $\widebar{A}=a_1, a_2, \ldots, a_p$, and $\widebar{B}=b_1, b_2, \ldots, b_q$, such that $A$ and $B$ are the respective sets of the two streams, and $n=\abs{A \cap B}$, find an estimate $\widehat{n}$ of $n$ using only $ m $ storage units, where $m \ll n$.

$\abs{A \cap B}$ can be estimated in a straightforward manner using the inclusion-exclusion principle: $\widehat{\abs{A \cap B}} = \widehat{\abs{A}} + \widehat{\abs{B}} - \widehat{\abs{A \cup B}} $, taking advantage of the fact that estimating $\abs{A}$, $\abs{B}$ and $\abs{A \cup B}$ is relatively easy.
However, we will later show that this scheme produces inaccurate results. In \cite{Beyer07}, it is proposed to estimate the Jaccard similarity, defined as
$\rho(A,B)=\frac{\abs{A \cap B}}{\abs{A \cup B}}$.
The idea is to estimate $\abs{A \cup B}$, and then to extract $\abs{A \cap B}$. A third scheme, proposed in \cite{Dasu02}, suggests that $\widehat{\abs{A \cap B}} = \widehat{\frac{\rho(A,B)}{\rho(A,B) + 1}}(\widehat{\abs{A}} + \widehat{\abs{B}})$.
With respect to the above three estimation schemes, the main contributions of this paper are as follows:
\begin{enumerate}
\item For the first time, we present a complete analysis of the statistical performance (bias and variance) of the above three schemes.
\item We find the optimal (minimum) variance of any unbiased set intersection estimator.
\item We present and analyze a new unbiased estimator, based on the Maximum Likelihood (ML) method, which outperforms the above three schemes.
\end{enumerate}

The rest of the paper is organized as follows. Section \ref{sec:related} discusses previous work and presents the three previously known schemes. 
Section \ref{sec:MLscheme} presents our new Maximum Likelihood (ML) estimator. It also shows that the new scheme achieves optimal variance and that it outperforms the three known schemes.
Section \ref{sec:analysis} analyzes the statistical performance (bias and variance) of the three known schemes. Section \ref{sec:sim} presents simulation results confirming that the new ML estimator outperforms the three known schemes.
Finally, Section \ref{sec:conclusion} concludes the paper.


\section{Related Work and Previous Schemes} \label{sec:related}

The database research community has extensively explored the problem of data cleaning: detecting and removing errors and inconsistencies from data to improve the quality of databases \cite{Rahm00}. Identifying which fields share similar values, identifying join paths, estimating join directions and sizes, and detecting inclusion dependencies are well-studied aspects of this problem \cite{Bauckmann07,Clifton97,Dasu02,Hernandez98,Monge00}.
For example, in \cite{Dasu02} the authors present several methods for finding related database fields. Their main idea is to hash the values of each field and keep a small sketch that contains the minimal hash values for each. Then, the Jaccard similarity is used to measure similarities between fields. 
In \cite{Bauckmann07}, the authors study the related problem of detecting inclusion dependencies, i.e., pairs of fields $A$ and $B$ such that $A \subseteq B$, and they present an efficient way to test all field pairs in parallel. 

All the above problems are closely related to the ``cardinality estimation problem" discussed in this paper.
This problem has received a great deal of attention in the past decade thanks to the growing number of important real-time ``big data" applications, such as estimating the propagation rate of viruses, detecting DDoS attacks \cite{Giroire2007, Ganguly2007}, and measuring general properties of network traffic \cite{Metwally:2008}.

Many works address the cardinality estimation problem \cite{Cosma:2011, CohenK08, Flajolet2007, Giroire2009, Lumbroso2010, Metwally:2008} and propose statistical algorithms for solving it. These algorithms are usually limited to performing only one pass on the received packets and using a fixed small amount of memory. A common approach is to hash every element into a low-dimensional data sketch, which can be viewed as a uniformly distributed random variable. Then, one of the following schemes is often used to estimate the number of distinct elements in the set:
\begin{enumerate}
\item Order-statistics based estimators: In this family of schemes, the identities of the smallest (or largest) $k$ elements are remembered for the considered set. These values are then used for estimating the total number of distinct elements \cite{Beyer07, CohenK08, Giroire2009, Lumbroso2010}. The family of estimators with $k=1$ (where the minimal/maximal identity is remembered) is also known as min/max sketches. 
\item Bit-pattern based estimators: In this family of schemes, the highest position of the leftmost 1-bit in the binary representation of the identity of each element is remembered and then used for the estimation \cite{Cosma:2011, Flajolet2007}.
\end{enumerate}
If only one hash function is used, the schemes estimate the value of $n$ with an infinite variance. To bound the variance, both schemes repeat the above procedures for $m$ different hash functions and use their combined statistics for the estimation\footnote{Stochastic averaging can be used to reduce the number of hash functions from $m$ to only two \cite{Flajolet:1985}.}.

A comprehensive overview of different cardinality estimation techniques is given in \cite{Cosma:2011, Metwally:2008}. State-of-the art cardinality estimators have a standard error of about $ 1/\sqrt{m} $, where $ m $ is the number of storage units \cite{Chassaing2006}. The best known cardinality estimator is the HyperLogLog algorithm \cite{Flajolet2007}, which belongs to the family of min/max sketches and has a standard error of $ 1.04 / \; \sqrt{m} $, \cite{Flajolet2007}. For instance, this algorithm estimates the cardinality of a set with $ 10^{9} $ elements with a standard error of $ 2 \% $ using $ m=2,048 $ storage units.

Cardinality estimation algorithms can be used for estimating the cardinality of set intersection. As mentioned in Section \ref{sec:intro}, a straightforward technique is to estimate the intersection using the following inclusion-exclusion principle:
\begin{equation*} 
\widehat{\abs{A \cap B}} = \widehat{\abs{A}} + \widehat{\abs{B}} - \widehat{\abs{A \cup B}} \text{.}
\end{equation*}
This method will be referred to as Scheme-1.
Other algorithms first estimate the Jaccard similarity $\rho(A,B)=\frac{\abs{A \cap B}}{\abs{A \cup B}}$, and then use some algebraic manipulation on it \cite{Beyer07, Dasu02}.
Specifically, the scheme proposed in \cite{Beyer07} estimates both the Jaccard similarity and $\abs{A\cup B}$, and then uses
\begin{equation*} 
\widehat{\abs{A \cap B}} = \widehat{\rho(A,B)}\cdot \widehat{\abs{A \cup B}} \text{.}
\end{equation*}
This scheme will be referred to as Scheme-2.
The third scheme discussed in this paper, referred to as Scheme-3, is presented in \cite{Dasu02}. It estimates the Jaccard similarity, $\abs{A}$, and $\abs{B}$, and then uses 
\begin{equation*} 
\widehat{\abs{A \cap B}} = \widehat{\frac{\rho(A,B)}{\rho(A,B) + 1}}(\widehat{\abs{A}} + \widehat{\abs{B}}) \text{.}
\end{equation*}
The above equation is obtained by substituting the Jaccard similarity definition into $\frac{\rho(A,B)}{\rho(A,B)+1}$, which yields that
\begin{equation*}
\frac{\rho(A,B)}{\rho(A,B) + 1} = \frac{\abs{A \cap B}}{\abs{A \cap B} + \abs{A \cup B}} = \frac{\abs{A \cap B}}{\abs{A} + \abs{B}} \text{.}
\end{equation*}

In \cite{Bauckmann07, Kohler10}, the set intersection estimation problem is solved using smaller sample sets. While these techniques are simple and unbiased, they are inaccurate for small sets. In addition, they are sensitive to the arrival order of the data, and to the repetition pattern.


\section{A New Maximum Likelihood Scheme with Optimal Variance} \label{sec:MLscheme}

In this section we present a new unbiased estimator for the set intersection estimation problem. Because this estimator is based on the Maximum Likelihood (ML) method, it achieves optimal variance and outperforms the three known schemes.

Maximum-Likelihood estimation (ML) is a method for estimating the parameters of a statistical model. 
For example, suppose we are interested in the height distribution of a given population, but are unable to measure the height of every single person. Assuming that the heights are Gaussian distributed with some unknown mean and variance, the mean and variance can be estimated using ML and only a small sample of the overall population.
In general, for a given set of data samples and an underlying statistical model, ML finds the values of the model parameters that maximize the likelihood function, namely, the ``agreement" of the selected model with the given sample.

In the new scheme, we first find the (probability density) likelihood function of the set intersection estimation problem, $L(x_A=s,x_B=t ; \theta)$; namely, given $\theta=(a,b,n)$ as the problem parameters\footnote{The set intersection estimation problem has six identifying parameters ($n,u,a,b,\alpha,\beta$), only three of which are needed to derive the others.}, this is the probability density of $s$ to be the maximal hash value for $A$ and $t$ to be the maximal hash value for $B$. Then, we look for values of the problem parameters that maximize the likelihood function.

Table \ref{table:notations} shows some of the notations we use for the rest of the paper.

\begin{table}[htb!]
\centering 
	\begin{tabular}{|c|c|} \hline
   value & notation   \\\hline
	$\abs{A \cap B}$ & $n$  \\\hline
	$\abs{A \cup B}$ & $u$  \\\hline
	$\abs{A}$ & $a$  \\\hline
	$\abs{B}$ & $b$  \\\hline
	$\abs{A \setminus B}$ & $\alpha$  \\\hline
	$\abs{B \setminus A}$ & $\beta$   \\\hline
	\end{tabular}
	\caption{Notations} 
	\label{table:notations}
\end{table}

\subsection{The Likelihood Function of the Set Intersection Estimation Problem} \label{sub:likelihood}

We first find the likelihood function for one hash function $h_k$, namely, $L(x_A=s,x_B=t ; \theta)_k$, and then generalize it for all hash functions.
Recall that for the $k$th hash function, $x_A^k=\max_{i=1}^{a}\set{h_k(a_i)}$ and $x_B^k=\max_{j=1}^{b}\set{h_k(b_j)}$. To simplify the notation, we shall omit the superscript $k$ for $x_A^k$ and $x_B^k$, and use $x_A$ and $x_B$ respectively.

We use $\text{PDF}_U(w)$ to denote the probability density function (PDF) of a uniformly distributed random variable $\Unif{0}{1}$ at $w$.
Denote the elements in $A\cap B$ as $\set{z_1,z_2,\ldots,z_n}$. Thus, the elements in $A$ and the elements in $B$ can be written as $\set{x_1,x_2,\ldots,x_{\alpha}, z_1,z_2,\ldots,z_n}$ and $\set{y_1,y_2,\ldots,y_{\beta}, z_1,z_2,\ldots,z_n}$ respectively (see Table \ref{table:notations}).

We now divide the likelihood function according to the three possible relations between $x_A$ and $x_B$: $x_A=x_B$, $x_A>x_B$ and $x_A<x_B$.

\begin{description}
    \item[\textnormal{Case 1:} $x_A = x_B$] \hfill \smallskip

When $x_A=x_B=s$ holds, the element with the maximal hash value must belong to $A \cap B$. The likelihood function of $\theta$ given this outcome is \begin{align}
\label{eq:case1}
L(x_A=x_B=s ; \theta) &= \sum_{i=1}^{n}{\text{PDF}_U(s) \cdot \Prob{x_{(A\cup B) \setminus {\{z_i\}}} < s}} \nonumber \\
&=\sum_{i=1}^{n}s^{u-1} = n\cdot s^{u-1}
\text{.}
\end{align}
This equality holds because there are $n$ possible elements in $A\cap B$ whose hash value can be the maximum in $A\cup B$, and because $\text{PDF}_U(s)=1$.

    \item[\textnormal{Case 2:} $x_A < x_B$] \hfill \smallskip  

In order to have $x_A < x_B$, where $x_A=s$ and $x_B=t$, the maximal hash value in $B$ must also be in $B\setminus A$, and its value must be $t$. The likelihood function of $\theta$ in this case is
\begin{equation*}
L(x_{B\setminus A}=t ; \theta) = \sum_{j=1}^{\beta} {\text{PDF}_U(t) \cdot \Prob{x_{(B\setminus A) \setminus {\{y_j\}}} < t}} = \beta \cdot t^{\beta -1} \text{.}
\end{equation*}

In addition, the maximal hash value in $A$ must be $s$. The probability density for this is
\begin{equation*}
L(x_A=s ; \theta) = \sum_{e \in A} {\text{PDF}_U(s) \cdot \Prob{x_{A \setminus {\{e\}}} < s}} = a s^{a-1} \text{.}
\end{equation*}

Thus, 
\begin{equation} 
\label{eq:case2}
L(x_A < x_B , x_A=s , x_B=t ; \theta) = a s^{a - 1} \cdot \beta t^{\beta - 1} \text{.}
\end{equation}

	\item[\textnormal{Case 3:} $x_A > x_B$] \hfill \smallskip

This case is symmetrical to the previous case. Thus, the likelihood function of $\theta$ in this case is
\begin{equation*}
L(x_A>x_B , x_A=s , x_B=t ; \theta) = \alpha s^{\alpha - 1} b t^{b - 1} \text{.}
\end{equation*}
\end{description}

Thus, the likelihood function for set intersection is
\begin{equation*}
L(x_A=s,x_B=t ; \theta)_k = 
\begin{cases} 
      n s^{u - 1} & x_A = x_B \\
      \beta t^{\beta - 1} a s^{a - 1} & x_A < x_B \\
      \alpha s^{\alpha - 1} b t^{b - 1} & x_A > x_B \: \text{.}
   \end{cases}
\end{equation*}
We now use the following indicator variables:
\begin{enumerate}
\item $I_{1} = 1$ if $x_A=x_B$, and $I_{1} = 0$ otherwise,
\item $I_{2} = 1$ if $x_A<x_B$, and $I_{2} = 0$ otherwise,
\item $I_{3} = 1$ if $x_A>x_B$, and $I_{3} = 0$ otherwise,
\end{enumerate}
to obtain that
\begin{equation}\label{eq:hkhash}
L(x_A=s,x_B=t ; \theta)_k = (n s^{u - 1}) ^ {I_{1}} 
\cdot (\beta t^{\beta - 1} a s^{a - 1}) ^ {I_{2}}
\cdot (\alpha s^{\alpha - 1} b t^{b - 1}) ^ {I_{3}} \text{.}
\end{equation}

Eq. (\ref{eq:hkhash}) states the likelihood function for one hash function. To generalize this equation to all $m$ hash functions, denote $S=(s_1,s_2,\ldots,s_m)$ and $T=(t_1,t_2,\ldots,t_m)$ as the $m$-dimensional vectors of $s$ and $t$ for each hash function.
\begin{corollary}\label{cor:fullLikeFun} \ \\
The likelihood function for the set intersection estimation problem, for all $m$ hash functions, satisfies:
\begin{equation*} 
L(x_A=S,x_B=T ; \theta) = \prod_{k=1}^{m}(n (s_k)^{u - 1}) ^ {I_{1,k}} 
\cdot (\beta (t_k)^{\beta - 1} a (s_k)^{a - 1}) ^ {I_{2,k}}
\cdot (\alpha (s_k)^{\alpha - 1} b (t_k)^{b - 1}) ^ {I_{3,k}} \text{,}
\end{equation*}
where $I_{1,k}$ is the value of $I_{1}$ for hash function $k$, and the same holds for $I_{2,k}$ and $I_{3,k}$.
\QEDA
\end{corollary}

It is usually easier to deal with the log of a likelihood function than with the likelihood function itself.
Because the logarithm is a monotonically increasing function, its maximum value is obtained at the same point as the maximum of the function itself. 
In our case,
\begin{align} \label{eq:loglike}
\log{L(x_A=S,x_B=T ; \theta)} &= \log{\prod_{k=1}^{m}L(x_A^k=s_k,x_B^k=t_k ; \theta)_k} 
\nonumber \\ &= \sum_{k=1}^{m}{\log{L(x_A^k=s_k,x_B^k=t_k ; \theta)_k}} \nonumber \\
&= \sum_{k=1}^{m}{I_{1,k} \cdot \log{(n \cdot (s_k)^{u - 1})}} 
+ \sum_{k=1}^{m}{I_{2,k} \cdot \log{(\beta \cdot (t_k)^{\beta - 1} \cdot a \cdot (s_k)^{a - 1})}} \nonumber \\
&+ \sum_{k=1}^{m}{I_{3,k} \cdot \log{(\alpha \cdot (s_k)^{\alpha - 1} \cdot b \cdot (t_k)^{b - 1})}} \text{.}
\end{align}

\subsection{The New Scheme} \label{sub:ml}

We use Corollary \ref{cor:fullLikeFun} in order to find $\theta=(a,b,n)$ that maximizes Eq. (\ref{eq:loglike}). 
Let $g(a,b,n)$ and $\mathbb{H}(a,b,n)$ be the gradient and the Hessian matrix of the log-likelihood function.
Namely, $g$ is the vector whose components are the partial derivatives of the log-likelihood function for the problem parameters $\theta=(a,b,n)$:
\begin{equation}\label{eq:g}
g(a,b,n) = (\frac{\partial \log L}{\partial a}, \frac{\partial \log L}{\partial b}, \frac{\partial \log L}{\partial n}) \text{,}
\end{equation}
and $\mathbb{H}$ is the matrix of the second-order partial derivatives of the log-likelihood function
\begin{equation} \label{eq:h}
\mathbb{H}_{i,j} = \frac{\partial^2}{\partial \theta_i \partial \theta_j} \log L \:\:\:\:\: , \: 1\leq i,j \leq 3 \: \text{,}
\end{equation}
where $\theta_1=a$, $\theta_2=b$ and $\theta_3=n$.

The new scheme finds the maximal value of the log-likelihood function, i.e., the root of its gradient $g(a,b,n)$. Practically, this is done using iterations of the Newton-Raphson method on the gradient and Hessian matrix $g$ and $H$. Starting from an initial estimation $\widehat{\theta_0}=(\widehat{a_0},\widehat{b_0},\widehat{n_0})$, the Newton-Raphson method implies that a better estimation is
$\widehat{\theta_1} = \widehat{\theta_0} - H^{-1}(\widehat{\theta_0}) \cdot g(\widehat{\theta_0})$.
The process is repeated, namely,
\begin{equation}\label{eq:nr}
\widehat{\theta_{l+1}} = \widehat{\theta_l} - H^{-1}(\widehat{\theta_l}) \cdot g(\widehat{\theta_l}) \text{,}
\end{equation} 
until a sufficiently accurate estimation is reached. This idea is summarized in the following algorithm.

\begin{algorithm2}
\textbf{\newline (A Maximum Likelihood scheme for the set intersection estimation problem)}
\label{alg:MLscheme} \ \\
The scheme gets as an input the sketches of the sets $\set{x_A^k}_{k=1}^m$ and $\set{x_B^k}_{k=1}^m$, where $x_A^k$ and $x_B^k$ are the maximal hash values of $A$ and $B$ respectively for the $k$th hash function, and returns an estimate of their set intersection cardinality.
\begin{enumerate}
\item [1)] Estimate $a_0=\widehat{a}$, $b_0=\widehat{b}$ and $\widehat{u}$ using any cardinality estimation algorithm, such as \cite{Flajolet2007}.
\item [2)] Estimate the Jaccard similarity $\widehat{\rho}$ from the given sketches of $A$ and $B$.
\item [3)] Find the maximum of the likelihood function $L$ (Eq. (\ref{eq:loglike})) as explained above; use $n_0=\widehat{\rho}\cdot\widehat{u}$ as an initial value of $n$ (see Scheme-2 in Section 1), and $a_0,b_0$ as an initial values of $a$ and $b$ respectively.
\item [4)] Return $\widehat{n}$.
\end{enumerate}
\end{algorithm2}
When we implemented Algorithm \ref{alg:MLscheme}, we discovered that 3 Newton-Raphson iterations are enough for the algorithm to converge.

\subsection{The Optimal Variance of the New Estimator} \label{sub:optimal}

The new estimator proposed in this section is based on Maximum Likelihood and thus it asymptotically achieves optimal variance \cite{Shao1998}. We use the Cramer-Rao bound to compute this optimal variance.

The Cramer-Rao bound states that the inverse of the Fisher information matrix is a lower bound on the variance of any unbiased estimator of $\theta$ \cite{Shao1998}.
The Fisher information matrix $\mathbb{F}_{i,j}$ is a way of measuring the amount of information that a random variable $X$ carries about an unknown parameter $\theta$ upon which the probability of $X$ depends. It is defined as:
\begin{equation} \label{eq:fisher}
\mathbb{F}_{i,j} = -\Ex{\frac{\partial^2}{\partial \theta_i \partial \theta_j} \log L} \text{.}
\end{equation}

We now use the log-likelihood function (Eq. (\ref{eq:loglike})) to derive this matrix for the set intersection estimation problem:
\begin{equation} \label{eq:fisherMat}
\mathbb{F}(a,b,n) =
m \cdot \left( \begin{array}{ccc}
\frac{\beta}{u} \cdot \frac{1}{a^2} + \frac{1}{u \cdot \alpha} & 0 & \frac{-1}{u \cdot \alpha} \\
0 & \frac{\alpha}{u} \cdot \frac{1}{b^2} + \frac{1}{u \cdot \beta} & \frac{-1}{u \cdot \beta} \\
\frac{-1}{u \cdot \alpha} & \frac{-1}{u \cdot \beta} & \frac{1}{u \cdot n} + \frac{1}{u \cdot \beta} + \frac{1}{u \cdot \alpha} \end{array} \right) \text{,}
\end{equation}
where each term is derived due to algebraic manipulations and derivatives of the log-likelihood function. Note that the expected values of the indicator variables $I_{1,k},I_{2,k}$ and $I_{3,k}$ are required to derive the matrix. 
For $I_{1,k}$ we get:
\begin{equation*}
\Ex{I_{1,k}} = \Prob{x_A^k=x_B^k} = \frac{n}{u} \:\:\:\:\: , \: 1\leq k \leq m \: \text{.}
\end{equation*}
The first equality is due to the definition of $I_{1,k}$, and the second is due to Eq. (\ref{eq:pro}). The following are obtained in the same way for every $1\leq k \leq m$:
\begin{enumerate}
\item $\Ex{I_{2,k}} = \Prob{x_A^k<x_B^k}=\frac{\beta}{u}$.
\item $\Ex{I_{3,k}} = \Prob{x_A^k>x_B^k}=\frac{\alpha}{u}$.
\end{enumerate}

Let $\widehat{\abs{A \cap B}}$ be an unbiased estimator for the set intersection estimation problem. 
Then, according to the Cramer-Rao bound, $\Var{\widehat{\abs{A \cap B}}} \geq (\mathbb{F}^{-1})_{3,3}$, where $(\mathbb{F}^{-1})_{3,3}$ is the term in place $[3,3]$ in the inverse Fisher information matrix. Finally, from the computation of the term ($\mathbb{F}^{-1})_{3,3}$, we can obtain the following corollary:
\begin{corollary} \ \\
$\Var{\widehat{\abs{A \cap B}}} \geq (\mathbb{F}^{-1})_{3,3}$, where, \\

$(\mathbb{F}^{-1})_{3,3} = \frac{n \cdot u}{m} \cdot \frac{(b^2 + \alpha\beta)(a^2 + \alpha\beta)} {{\alpha\cdot n}(a^2 + \alpha\beta) + \beta \cdot n(b^2 + \alpha\beta) + (a^2 + \alpha\beta)(b^2 + \alpha\beta)}$.
\end{corollary}


\section{An Analysis of the Three Schemes From Section 1} \label{sec:analysis}

In this section we will analyze the statistical performance (bias and variance) of the three schemes discussed in Section 1 for set intersection estimation.

\subsection{Preliminaries}

\subsubsection{Jaccard Similarity}
Recall that the Jaccard similarity is defined as: $\rho(A,B)=\frac{\abs{A \cap B}}{\abs{A \cup B}}$, where $A$ and $B$ are two finite sets. Its value ranges between $0$, when the two sets are completely different, and $1$, when the sets are identical.
An efficient and accurate estimate of $\rho$ can be computed as follows \cite{Broder97new}. 
Each item in $A$ and $B$ is hashed into $(0,1)$, and the maximal value of each set is considered as a sketch that represents the whole set. To improve accuracy, $m$ hash functions are used\footnote{Stochastic averaging can be used to reduce the number of hash functions from $m$ to only two \cite{Flajolet:1985}.}, and the sketch of each set is a vector of $m$ maximal values. Given a set $A = \set{a_1,a_2,\ldots,a_p}$ and $m$ different hash functions $h_1,h_2,\ldots,h_m$, the maximal hash value for the $j$th hash function can be formally expressed as:
\begin{equation*}
x_A^j=\max_{i=1}^{p}\set{h_j(a_i)} \:\:,\:\: 1\leq j \leq m \text{,} 
\end{equation*} 
and the sketch of $A$ is:
\begin{equation*}
X_{(A)}=\set{x_A^1,x_A^2,\ldots,x_A^m} \text{.}
\end{equation*} 
$X_{(B)}$ is computed in the same way. Then, the two sketches are used to estimate the Jaccard similarity of $A$ and $B$:
\begin{equation}
\label{eq:jacest}
\widehat{\rho(A,B)}=\frac{\sum_{j=1}^{m}I_{x_A^j==x_B^j}}{m} \text{,}
\end{equation}
where the indicator function $I_{x_A^j==x_B^j}$ is $1$ if $x_A^j=x_B^j$, and $0$ otherwise.
To shorten our notation, for the rest of the paper we use $\rho$ to indicate $\rho(A,B)$.

\begin{lemma} \label{lemma:jacard} \ \\
In Eq. (\ref{eq:jacest}), $\widehat{\rho}$ is a normally distributed random variable with mean $\rho$ and variance $\frac{1}{m}\rho(1-\rho)$; i.e., $\widehat{\rho} \to \Normal{\rho}{\frac{1}{m}\rho(1-\rho)}$.
\end{lemma}
\begin{proof} \ \\
Consider the $j$-th hash function. According to \cite{Broder97new},
\begin{equation}
\label{eq:pro}
\Prob{x_A^j=x_B^j}=\frac{\abs{A \cap B}}{\abs{A \cup B}} \text{.}
\end{equation}
The intuition is to consider the hash function $h_j$ and define $m(S)$, for every set $S$, to be the element in $S$ with the maximum hash value of $h_j$, i.e., $h_j(m(S)) = x_S^j$. Then, we get $m(A)=m(B)$ only when $m(A \cup B)$ lies also in their intersection $A \cap B$. The probability of this is the Jaccard ratio $\rho$, and therefore $\Prob{x_A^j = x_B^j} = \rho$.

From Eqs. (\ref{eq:jacest}) and (\ref{eq:pro}) follows that $\widehat{\rho}$ is a sum of $m$ Bernoulli variables. Therefore, it is binomially distributed, and can be asymptotically approximated to normal distribution as $m\to \infty$; namely, $\widehat{\rho} = \frac{\sum_{l=1}^{m}I_{x_A^j=x_B^j}}{m} \to \Normal{\rho}{\frac{1}{m}\rho(1-\rho)}$.
\end{proof}

\subsubsection{The Cardinality Estimation Problem}

Algorithms for estimating the cardinality of set intersection use estimations of $\abs{A}$, $\abs{B}$, and $\abs{A \cup B}$. These estimations can be found using well-known algorithms for the following cardinality estimation problem:

\refstepcounter{problemC}
\begin{description}\label{prob:cardinality}
    \item[Instance:] A stream of elements $x_1, x_2, \ldots, x_s$ with repetitions. Let $c$ be the number of different elements, namely $c=\abs{\set{x_1, x_2, \ldots, x_s}}$.
    \item[Objective:] Find an estimate $\widehat{c}$ of $c$ using only $ m $ storage units, where $m \ll c$.
\end{description}

For the rest of the paper we consider the HyperLogLog algorithm \cite{Flajolet2007} for solving the above problem. 
As indicated in Section \ref{sec:related}, this algorithm has a very small standard error, of about $ 1.04 / \; \sqrt{m} $ where $ m $ is the number of storage units. The pseudo-code of this algorithm is as follows:

\begin{algorithm2} 
{The HyperLogLog algorithm for the cardinality estimation problem}
\label{alg:hyperLogLog}
\begin{enumerate}
\item Initialize $m$ registers: $C_1,C_2,\ldots,C_m$ to 0.
\item For each input element $x_i$ do:
\begin{enumerate}
\item Let $\rho= \floor{-\log_2\paren{h_1(x_i)}}$ be the leftmost 1-bit position of the hashed value.
\item Let $j=h_2(x_i)$ be the bucket for this element.
\item $C_j \leftarrow \max{\set{C_j,\rho}}$.
\end{enumerate}
\item To estimate the value of $n$ do:
\begin{enumerate}
\item $Z \leftarrow {(\sum_{j=1}^{m}2^{-C_j})}^{-1}$ is the harmonic mean of $2^{C_j}$.
\item return $\alpha_m m^2 Z$, where \\
$\alpha_m= \paren{m \int^{\infty}_{0}\paren{\log_2{\paren{\frac{2+u}{1+u}}}}^m \,du }^{-1}$.
\end{enumerate}
\end{enumerate}
\end{algorithm2}

The following lemma summarizes the statistical performance of Algorithm \ref{alg:hyperLogLog}:

\begin{lemma} \label{lemma:hyper} \ \\
For Algorithm \ref{alg:hyperLogLog}, 
$\widehat{c} \to \Normal{c}{\frac{c^2}{m}}$, where $c$ is the actual cardinality of the considered set, $\widehat{c}$ is the estimate computed by the algorithm, and $m$ is the number of storage units used by the algorithm. When Algorithm \ref{alg:hyperLogLog} is used with two sets $A$ and $B$, the following holds:
\begin{equation*} 
\widehat{\abs{A}} \to \Normal{\abs{A}}{\frac{\abs{A}^2}{m}} \text{,}
\end{equation*}
\begin{equation*}
\widehat{\abs{B}} \to \Normal{\abs{B}}{\frac{\abs{B}^2}{m}} \text{,}
\end{equation*}
\begin{equation} \label{eq:union}
\text{and} \:\:\:\:\: \widehat{\abs{A \cup B}} \to \Normal{\abs{A \cup B}}{\frac{\abs{A \cup B}^2}{m}} \text{.}
\end{equation}
\end{lemma}
The proof is given in \cite{Flajolet2007}.

Let us also recall three general lemmas, not related to set intersection cardinality estimation. The first lemma, known as the Delta Method, allows us to compute the probability distribution for a function of an asymptotically normal estimator using the estimator's variance:
\begin{lemma}[Delta Method] \label{lemma:delta} \ \\
Let $\theta_1,\theta_2,\ldots,\theta_m$ be a sequence of $m$ random variables such that for every integer $i$ $\sqrt{i}(\theta_i - \theta) \to \Normal{0}{\sigma^2}$, where $\theta$ and $\sigma^2$ are finite valued constants. Then, for every integer $i$ and for every function $g$ for which $g^\prime(\theta)$ exists and $g^\prime(\theta) \neq 0$, the following holds:
\begin{equation*}
\sqrt{i}(g(\theta_i) - g(\theta)) \to \Normal{0}{\sigma^2 {g^\prime(\theta)}^2}\text{.}
\end{equation*}
\end{lemma} 
A proof is given in \cite{Shao1998}.

The next lemma shows how to compute the probability distribution of a random variable that is a product of two normally distributed random variables whose covariance is $0$:
\begin{lemma}[Product distribution]\label{lemma:product} \ \\
Let $X$ and $Y$ be two random variables satisfying $X \to \Normal{\mu_x}{\sigma_x^2}$ and $Y \to \Normal{\mu_y}{\sigma_y^2}$, such that $\Cov{X}{Y}=0$.
Then, the product $X \cdot Y$ asymptotically satisfies the following:
\begin{equation*}
X\cdot Y \to \Normal{\mu_x\mu_y}{\mu_y^2\sigma_x^2 + \mu_x^2\sigma_y^2} \text{.}
\end{equation*}
\end{lemma} 
A proof is given in \cite{Shao1998}.

The final lemma states the distribution of the maximal hash value.
Let us first recall the beta distribution.
$\Beta{\alpha}{\beta}$ is defined over the interval $(0,1)$ and has the following probability and cumulative density functions (PDF and CDF respectively):
\begin{equation*}\label{eq:pdf}
	f(x) = \frac{\Gamma \paren{\alpha+\beta}}{\Gamma \paren{\beta}\Gamma \paren{\alpha}} x^{\alpha-1} (1-x)^{\beta-1}
\end{equation*}
\begin{equation*}\label{eq:cdf}
	F(x) = \int_0^{x}{\frac{\Gamma \paren{\alpha+\beta}}{\Gamma \paren{\beta}\Gamma \paren{\alpha}} x^{\alpha-1} (1-x)^{\beta-1}}dx \text{,}
\end{equation*}
where $ \Gamma(z) $ is the gamma function, defined as $\int_{0}^{\infty}{e^{-t}t^{z-1} dt}$. Using integration by parts, the gamma function can be shown to satisfy $\Gamma(z+1) = z \cdot \Gamma(z)$. 
Combining this with $\Gamma(1) = 1$ yields that $\Gamma(n) = (n-1)!$ holds for every integer $n$.
Two other known beta identities are \cite{Krishnamoorthy2006}:
\begin{equation*}\label{eq:beta11}
    \Beta{1}{1} \sim \Unif{0}{1}
\end{equation*}
and
\begin{equation*}\label{eq:betaFlip}
     \Beta{\alpha}{\beta} \sim 1- \Beta{\beta}{\alpha} \text{.}
\end{equation*}

The following lemma presents some key properties of the beta distribution, which we will use in the analysis.
\begin{lemma} \label{lemma:maxHash} \ \\
Let $x_1,x_2, \ldots ,x_n$ be independent RVs, where $x_i\sim\Unif{0}{1}$. Then,
\begin{enumerate}
\item[(a)] $X=\max_{i=1}^{n}{x_i}\sim\Beta{n}{1}$.
\item[(b)] $X$ satisfies the following
\begin{enumerate}
\item[(1)] $\Ex{X} = \frac{n}{n+1}$; and
\item[(2)] $\Var{X}= \frac{n}{(n+1)^2(n+2)}$.
\end{enumerate}
\end{enumerate}
\end{lemma}
A proof for (a) is given in \cite{CYK}; the other equalities follow the beta distribution of $X$.

\subsection{Analysis of Scheme-1}
Scheme-1 estimates the cardinality of $A \cap B$ using the inclusion-exclusion principle:
\begin{equation*} 
\widehat{\abs{A \cap B}} = \widehat{\abs{A}} + \widehat{\abs{B}} - \widehat{\abs{A \cup B}} \text{.}
\end{equation*}
Let $\widehat{n_{1}}$ be the estimator found by Scheme-1. 
The following theorem summarizes its statistical performance.

\begin{theorem} \label{thm:n1} \ \\
$\frac{\widehat{n_{1}}}{n} \to \Normal{1}{\frac{1}{m n^2}(u^2 - a^2 - b^2) - \frac{2a\cdot b}{m \cdot u \cdot n} + \frac{2 u \cdot (a^2(b^2 + \alpha\beta) + b^2(a^2 + \alpha\beta))}{m\cdot Z \cdot n}}$, where $m$ is the number of storage units, and $Z$ satisfies:
\begin{equation*}
Z = {\alpha\cdot n}(a^2 + \alpha\beta) + \beta \cdot n(b^2 + \alpha\beta) + (a^2 + \alpha\beta)(b^2 + \alpha\beta)\text{.}
\end{equation*}
\end{theorem}
\begin{proof} \ \\
For the expectation, from Lemma \ref{lemma:hyper} follows that:
\begin{align*}
\Ex{\widehat{n_1}} &= \Ex{\widehat{a}} + \Ex{\widehat{b}} - \Ex{\widehat{u}} = \\
&= a + b - u = n \text{.}
\end{align*}
The first equality is due to the definition of Scheme-1 and the expectation properties, and the second equality is due to Lemma \ref{lemma:hyper}. Thus, the estimator is unbiased. For the variance, Lemma \ref{lem:covAandB} in the Appendix proves that $\Cov{a}{b} = \frac{n \cdot a \cdot b}{m \cdot u}$, and Lemma \ref{lem:covAandU} in the Appendix proves that $\Cov{a}{u} = \frac{a^2}{m} + \frac{n \cdot a \cdot b}{ m \cdot u} - \frac{u \cdot n \cdot a^2(b^2 + \alpha\beta)}{m\cdot Z}$ and $\Cov{b}{u} = \frac{b^2}{m} + \frac{n \cdot a \cdot b}{ m \cdot u} - \frac{u \cdot n \cdot b^2(a^2 + \alpha\beta)}{m \cdot Z}$, where $Z = {\alpha\cdot n}(a^2 + \alpha\beta) + \beta \cdot n(b^2 + \alpha\beta) + (a^2 + \alpha\beta)(b^2 + \alpha\beta)$. We get that:
\begin{align}
\Var{\widehat{n_1}} &= \Var{\widehat{a}} + \Var{\widehat{b}} + \Var{\widehat{u}} + 2\Cov{\widehat{a}}{\widehat{b}} - 2\Cov{\widehat{a}}{\widehat{u}} - 2\Cov{\widehat{b}}{\widehat{u}} = \nonumber \\ 
&= \frac{a^2}{m} + \frac{b^2}{m} + \frac{u^2}{m} + 2\cdot\frac{n \cdot a \cdot b}{m \cdot u} \nonumber \\ 
&- 2 \cdot \frac{a^2}{m} - 2\cdot \frac{n \cdot a \cdot b}{ m \cdot u} + 2 \cdot\frac{u \cdot n \cdot a^2(b^2 + \alpha\beta)}{m\cdot Z} \nonumber \\
&- 2 \cdot \frac{b^2}{m} - 2\cdot \frac{n \cdot a \cdot b}{ m \cdot u} + 2\cdot\frac{u \cdot n \cdot b^2(a^2 + \alpha\beta)}{m \cdot Z} = \nonumber \\
&= \frac{1}{m}(u^2 - a^2 - b^2) - \frac{2n \cdot a \cdot b}{m \cdot u} + \frac{2 u \cdot n \cdot (a^2(b^2 + \alpha\beta) + b^2(a^2 + \alpha\beta))}{m\cdot Z} \text{.}
\end{align}
The first equality is due to variance properties and because $a$, $b$ and $u$ are dependent ($\Cov{a}{b}$, $\Cov{a}{u}$ and $\Cov{b}{u}$ are all $\neq 0$). The second equality is due to Lemma \ref{lemma:hyper}, Lemma \ref{lem:covAandB} and Lemma \ref{lem:covAandU} (both are in the Appendix). The third equality is due to algebraic manipulations. Finally, after dividing by $n^2$ we get the result.
\end{proof}

\subsection{Analysis of Scheme-2}
Scheme-2 estimates the cardinality of $A \cap B$ by estimating the Jaccard similarity $\rho$ and $\abs{A\cup B}$ \cite{Beyer07}:
\begin{equation*} 
\widehat{\abs{A \cap B}} = \widehat{\rho}\cdot \widehat{\abs{A \cup B}} \text{.}
\end{equation*}
Let $\widehat{n_{2}}$ be the estimator found by Scheme-2.
The following theorem summarizes its statistical performance.

\begin{theorem} \label{thm:n2} \ \\
$\frac{\widehat{n_{2}}}{n} \to \Normal{1}{\frac{1}{m\rho}}$, where $m$ is the number of storage units.
\end{theorem}
\begin{proof} \ \\
From the definition of Scheme-2, $\widehat{n_{2}}=\widehat{\rho}\cdot \widehat{u}$.
From Lemma \ref{lemma:jacard} and Lemma \ref{lemma:hyper} follows that:
\begin{enumerate}
\item $\widehat{\rho} \to \Normal{\rho}{\frac{1}{m}\rho(1-\rho)}$.
\item $\widehat{u} \to \Normal{u}{\frac{u^2}{m}}$.
\end{enumerate} 

Applying Lemma \ref{lemma:product} for the expectation yields:
\begin{equation*}
\Ex{\widehat{n_{2}}} = \Ex{\widehat{\rho}\cdot \widehat{u}} = \rho\cdot u = n \text{.}
\end{equation*}
Therefore, the estimator is unbiased. For the variance, applying again Lemma \ref{lemma:product} yields:
\begin{align*}
\Var{\widehat{n_{2}}} &= u^2\cdot\frac{1}{m}\rho(1-\rho) +  \rho^2\cdot\frac{u^2}{m} \\
&= \frac{1}{m}(u\cdot n - n^2 + n^2) \\
&= \frac{n^2}{m \rho} \text{,}
\end{align*}
where all the equalities are due to Lemma \ref{lemma:product} and algebraic manipulations. Finally, after dividing by $n^2$, we get $\frac{\widehat{n_{2}}}{n} \to \Normal{1}{\frac{1}{m \rho}}$.
\end{proof}

\subsection{Analysis of Scheme-3}
Scheme-3 estimates the cardinality of $A \cap B$ by estimating the Jaccard similarity $\rho$, $\abs{A}$, and $\abs{B}$ \cite{Dasu02}:
\begin{equation} 
\label{eq:sc3}
\widehat{\abs{A \cap B}} = \widehat{\frac{\rho}{\rho + 1}}(\widehat{\abs{A}} + \widehat{\abs{B}}) \text{.}
\end{equation}
Let $\widehat{n_{3}}$ be the estimator found by Scheme-3.
We will use the following lemma in the analysis.

\begin{lemma} \ \\
$\widehat{a} + \widehat{b} \to \Normal{a+b}{\frac{1}{m}(a^2 + b^2 + 2a b \rho)}$, where $m$ is the number of storage units.
\end{lemma}
\begin{proof}\ \\
For the expectation,
\begin{equation} \label{eq:exaplusb}
\Ex{\widehat{a}+\widehat{b}} = \Ex{\widehat{a}} + \Ex{\widehat{b}} = a + b \text{.}
\end{equation}
The first equality is due to expectation properties, and the second is due to Lemma \ref{lemma:hyper}. 
For the variance, Lemma \ref{lem:covAandB} in the Appendix proves that $\Cov{a}{b} = \frac{n \cdot a \cdot b}{m \cdot u}$. Thus,
\begin{align} \label{eq:varaplusb}
\Var{\widehat{a}+\widehat{b}} &= \nonumber \\ 
&= \Var{\widehat{a}} + \Var{\widehat{b}} + 2\Cov{\widehat{a}}{\widehat{b}} = \nonumber \\
&= \frac{a^2}{m} + \frac{b^2}{m} + 2\cdot\frac{n \cdot a \cdot b}{m \cdot u} = \nonumber \\
&= \frac{1}{m}(a^2+b^2+2a b \rho) \text{.}
\end{align}
The first equality is due to variance properties and because $a$ and $b$ are dependent ($\Cov{a}{b} \neq 0$). The second equality is due to Lemma \ref{lemma:hyper} and Lemma \ref{lem:covAandB}, and the third equality is due to algebraic manipulations and the Jaccard similarity definition ($\rho = \frac{n}{u}$).
Combining Eqs. \ref{eq:exaplusb} and \ref{eq:varaplusb} yields that
\begin{equation} \label{eq:term2}
\widehat{a} + \widehat{b} \to \Normal{a+b}{\frac{1}{m}(a^2 + b^2 + 2a b \rho)} \text{.}
\end{equation}
\end{proof}

The following theorem summarizes the statistical performance of $\widehat{n_3}$.

\begin{theorem} \label{thm:n3} \ \\
$\frac{\widehat{n_{3}}}{n} \to \Normal{1}{\frac{1}{m} (1 + \frac{2 a b}{u(a+b)} + (\alpha + \beta) \frac{u^2}{n (a+b)^2} )}$, where $m$ is the number of storage units.
\end{theorem}
\begin{proof} \ \\
From the definition of Scheme-3, $\widehat{n_{3}}=\widehat{\frac{\rho}{\rho + 1}}(\widehat{a} + \widehat{b})$.
From Lemma \ref{lemma:jacard} follows that:
\begin{equation*}
\widehat{\rho} \to \Normal{\rho}{\frac{1}{m}\rho(1-\rho)} \text{.}
\end{equation*}
Applying Lemma \ref{lemma:delta} on $\frac{\rho}{\rho+1}$ yields:
\begin{equation} \label{eq:term1}
\widehat{\frac{\rho}{\rho + 1}} \to \Normal{\frac{\rho}{\rho + 1}}{\frac{1}{m}\rho(1-\rho)\frac{1}{(1+\rho)^4}} \text{.}
\end{equation}
Applying Lemma \ref{lemma:product} for Eqs. (\ref{eq:term1}) and (\ref{eq:term2}) yields:
\begin{equation*}
\Ex{\widehat{n_{3}}} = \Ex{\widehat{\frac{\rho}{\rho + 1}}\cdot \widehat{a+b}} = \frac{\rho}{\rho + 1}\cdot(a+b) = n \text{.}
\end{equation*}
Therefore, the estimator is unbiased. For the variance, applying again Lemma \ref{lemma:product} yields:
\begin{align*}
\Var{\widehat{n_{3}}} = \frac{1}{m} (n^2 + \frac{2 a b {n}^2}{u(a+b)} + (\alpha + \beta) \frac{{u}^2 {n}}{(a+b)^2} ) \text{.}
\end{align*}
Finally, after dividing by $n^2$, we get:
$\frac{\widehat{n_{3}}}{n} \to \Normal{1}{\frac{1}{m} (1 + \frac{2 a b}{u(a+b) } + (\alpha + \beta) \frac{{u}^2}{n (a+b)^2} )}$.
\end{proof}

A simple comparison yields that $\Var{\widehat{n_{2}}} > \Var{\widehat{n_{3}}}$, i.e., Scheme-3 outperforms Scheme-2. However, a similar comparison between Scheme-1 and Scheme-3 cannot be done, because neither one is always better than the other.


\section{Simulation Results} \label{sec:sim}

In this section we examine the performance of our new ML estimator and show that it indeed outperforms the three known schemes. We implemented all four schemes, and simulated two sets, $A$ and $B$, whose cardinalities are as follows:
\begin{enumerate}
\item $\abs{A}=a=10^6$;
\item $\abs{B} = a \cdot f$, where $f>0$;
\item $\abs{A \cap B} = a \cdot \alpha$, where $0\leq\alpha \leq 1$.
\end{enumerate}

We estimate $\widehat{\abs{A \cap B}}$ for each of the four schemes, for $f \in \set{1,5,10}$, and for $\alpha \in \set{0,0.01,0.02,\ldots,0.98,0.99,1}$. We repeat the test for $10,000$ different sets. Thus, for each of the four schemes, and for each $f$ and $\alpha$ values, we get a vector of $10,000$ different estimations. Then, for each $f$ and $\alpha$ values, we compute the variance and bias of this vector, and view the result as the variance and bias of the estimator (for the specific $f$ and $\alpha$ values). Each such computation is represented by one point in the graph. Let $v_{f,\alpha}=(\widehat{n}_1,\ldots,\widehat{n}_{10^4})$ be the vector of estimations for a specific scheme and for specific $f$ and $\alpha$ values. Let $\mu = \frac{1}{10^4}\sum_{i=1}^{10^4}{\widehat{n}_i}$, be the mean of $v_{f,\alpha}$.
The bias and variance of $v_{f,\alpha}$ are computed as follows:
\begin{equation*}
\text{Bias}(v_{f,\alpha})=\abs{\frac{1}{n}(\mu - n)}
\end{equation*}
and
\begin{equation*}
\Var{v_{f,\alpha}}=\frac{1}{10^4}\sum_{i=1}^{10^4}{(\widehat{n}_i-\mu)^2}.
\end{equation*}

Figure \ref{fig:figBias} presents the bias of the ML estimator for $f=1$, different $\alpha$ values, and two values of $m$: $m=1,000$ and $m=10,000$ (recall that $m$ is the number of hash values used for the estimations of $\widehat{\abs{A}}$, $\widehat{\abs{B}}$ and the Jaccard similarity $\widehat{\rho}$).
We can see that the bias is very small for all $\alpha$ and $m$ values. We got very similar results for bigger $f$ values as well.

\begin{figure*}[tbp]
\begin{center}
\begin{tabular}{ccc}
\epsfxsize=0.45\textwidth  \epsffile{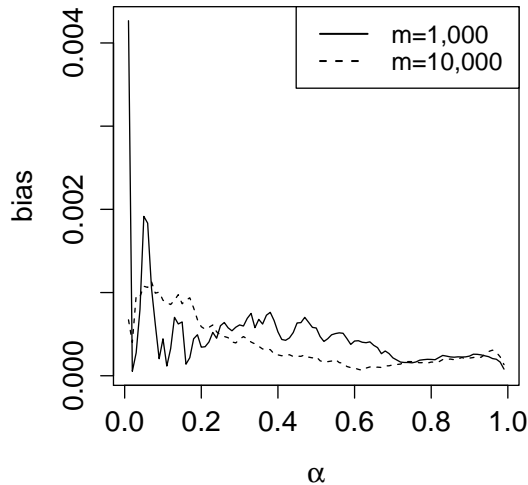}
\end{tabular}
\end{center}
\caption{The bias of the new ML scheme for $f=1$ and for different $\alpha$ values, for $m=1,000$ and $m=10,000$}
    \label{fig:figBias}
\end{figure*}

Figure \ref{fig:VarMLA} presents the normalized variance ($\Var{\frac{\widehat{n}}{n}}$) of our ML estimator for different $f$ and $\alpha$ values, and for $m \in \set{100,500,1000,10000}$. As expected, the normalized variance decreases as the number of hash values ($m$) increases, or as $\alpha$ increases. Overall, the normalized variance is very small for all values of $\alpha$, $f$ and $m$, indicating that the new scheme is very precise.

\begin{figure*}[tbp]
\begin{center}
\begin{tabular}{ccc}
\epsfxsize=0.3\textwidth  \epsffile{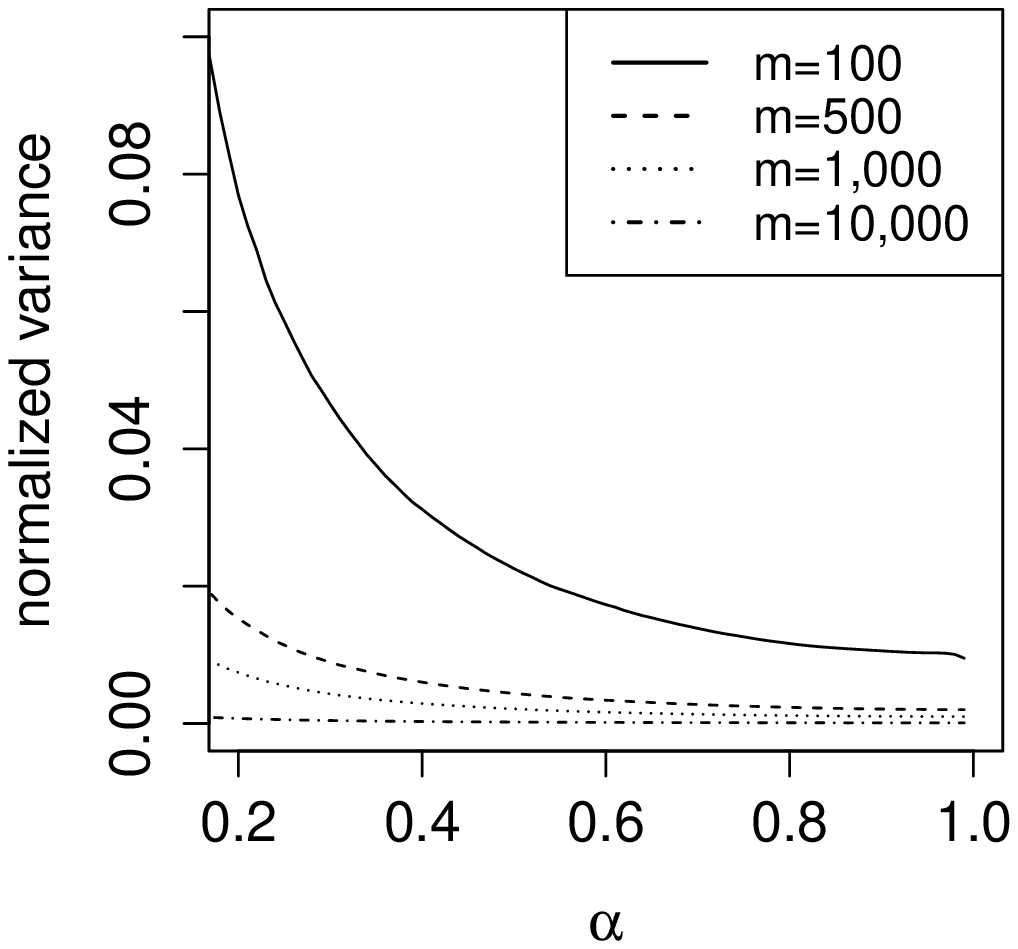}& \epsfxsize=0.3\textwidth  \epsffile{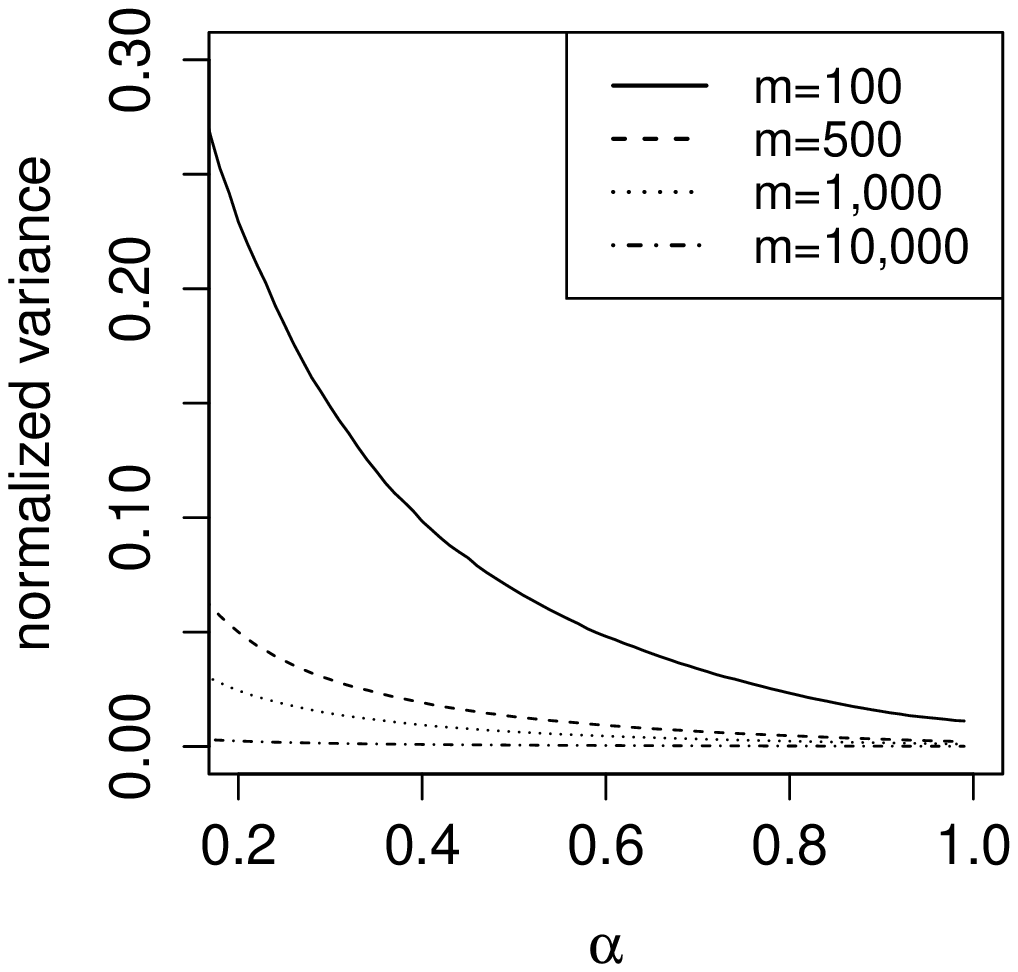} &
\epsfxsize=0.3\textwidth  \epsffile{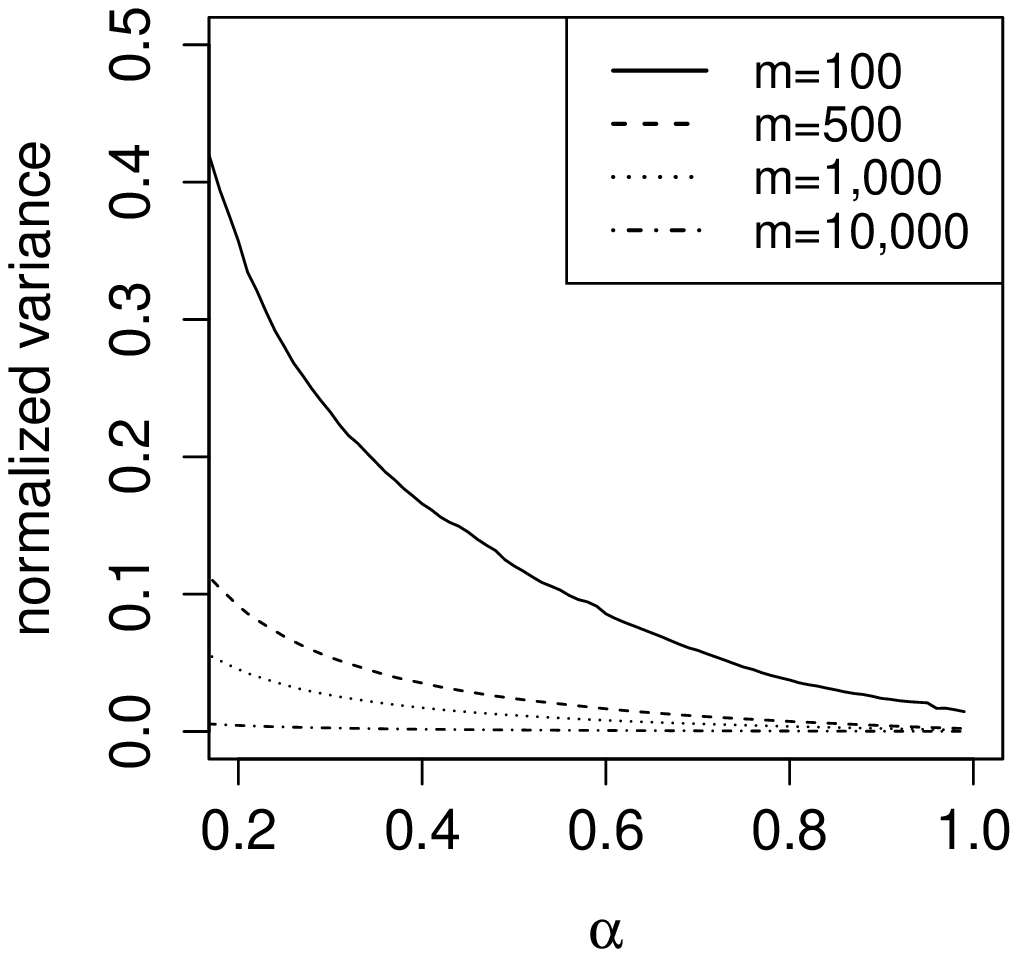}\\
{(a) $f=1$} & {(b) $f=5$} & {(c) $f=10$}
\end{tabular}
\end{center}
\caption{The normalized variance of the new ML scheme for different values of $f$, $\alpha$ and $m$}
    \label{fig:VarMLA}
\end{figure*}

After showing that the new ML scheme indeed yields good results, we now compare its performance to that of Schemes 1-3.
When comparing the statistical performance of two algorithms, it is common to look at their MSE (mean squared error) or RMSE, where $\text{MSE} = (\text{Bias}(\widehat{\theta}))^2 + \Var{\widehat{\theta}}$, and $\text{RMSE} = \sqrt{(\text{Bias}(\widehat{\theta}))^2 + \Var{\widehat{\theta}}}$. 
In our case, because all the estimators are unbiased, we compare only their variance. We define the ``relative variance improvement'' of the new ML scheme over each scheme as
\begin{equation*}
\frac{\Var{\widehat{\theta_i}} - \Var{\widehat{\theta_{\emph{ML}}}}}{\Var{\widehat{\theta_i}}} \text{,}
\end{equation*}
where $\theta_i$ is the estimator of the $i$th scheme, and $\theta_{\emph{ML}}$ is the new ML estimator.

Figure \ref{fig:fig1} presents the simulation results for two values of $m$: $m=10,000$ (upper graphs) and $m=1,000$ (lower graphs).
\textbf{We can see that the new ML estimator outperforms the three schemes for all values of $\alpha$ and $f$, and for both values of $m$.}
This improvement varies between $100 \%$ to a few percent.

\begin{figure*}[tbp]
\begin{center}
\begin{tabular}{ccc}
\epsfxsize=0.3\textwidth  \epsffile{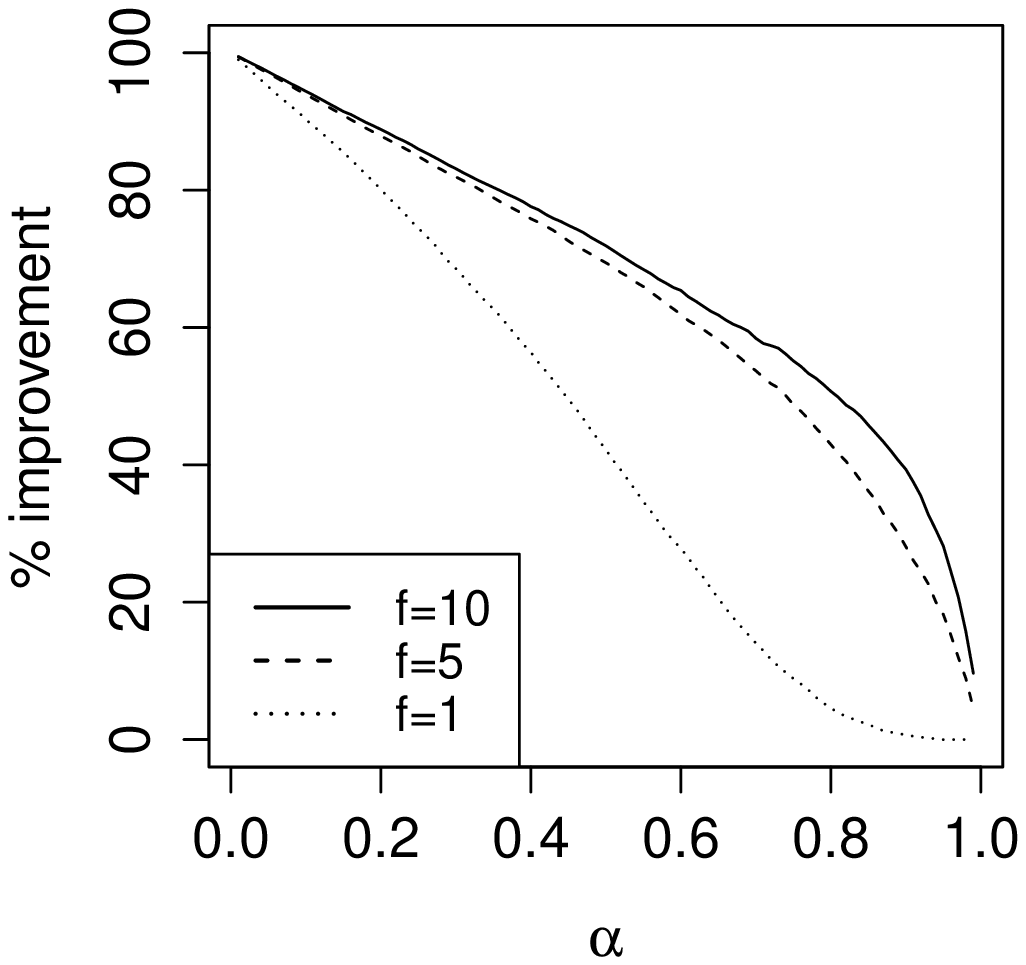}& \epsfxsize=0.3\textwidth  \epsffile{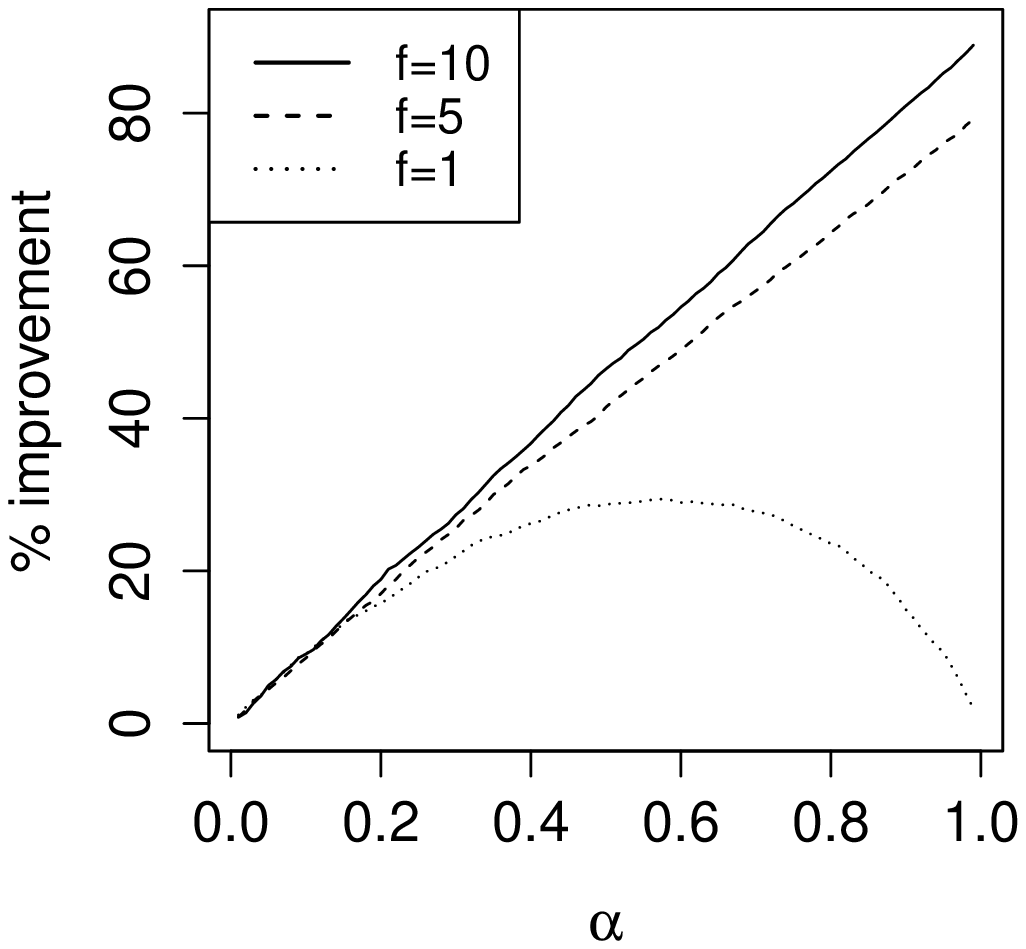} &
\epsfxsize=0.3\textwidth  \epsffile{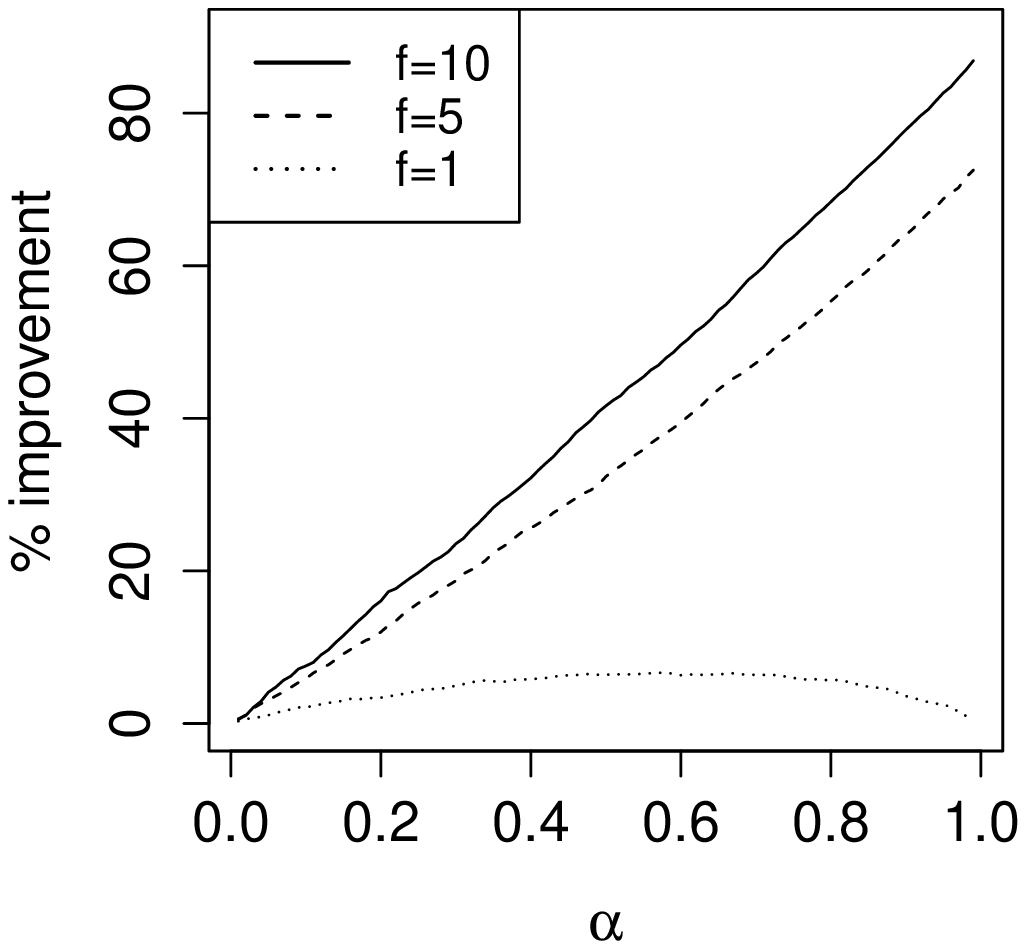}\\ 
\epsfxsize=0.3\textwidth  \epsffile{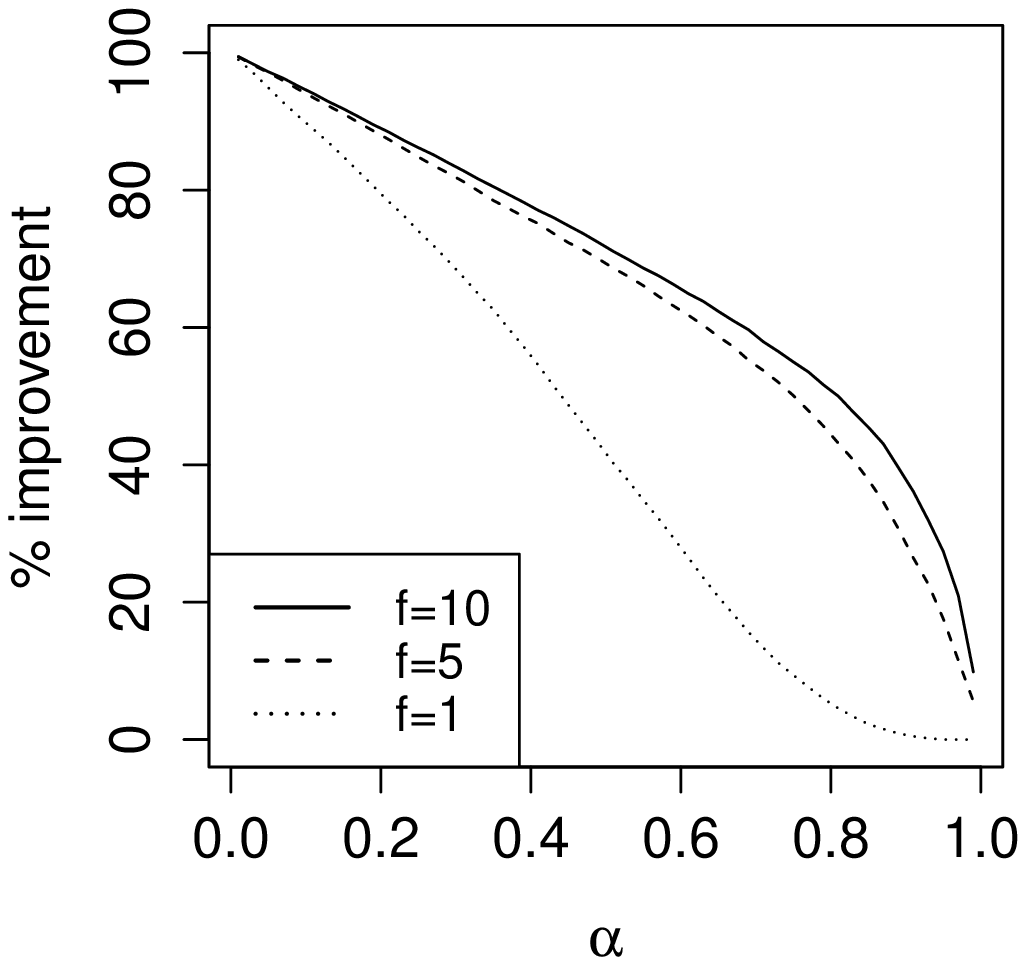}& \epsfxsize=0.3\textwidth  \epsffile{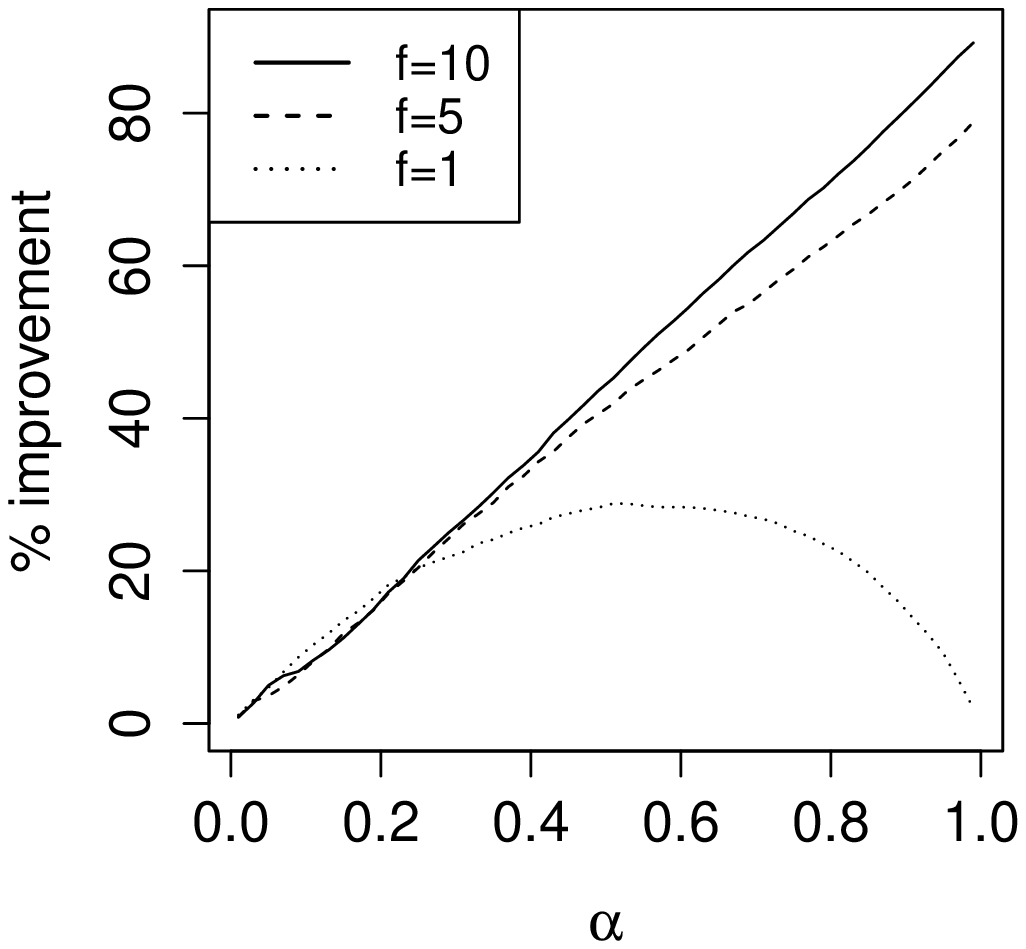} &
\epsfxsize=0.3\textwidth  \epsffile{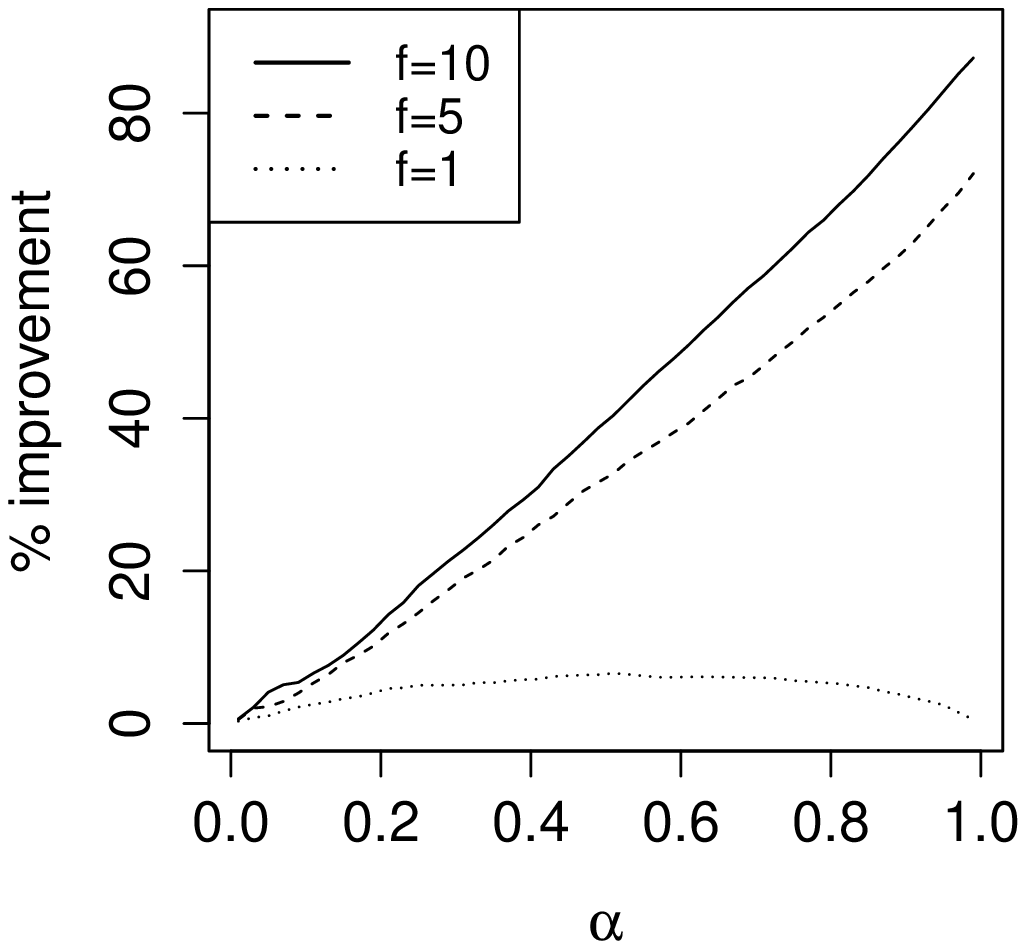}\\
{(a) Scheme-1} & {(b) Scheme-2} & {(c) Scheme-3}
\end{tabular}
\end{center}
\caption{The percentage of variance improvement of the new ML scheme over each of the other schemes for different values of $f$ and $\alpha$; the upper graphs are for $m=10,000$ and the lower graphs are for $m=1,000$}
    \label{fig:fig1}
\end{figure*}


\section{Conclusion} \label{sec:conclusion}
In this paper we studied the problem of estimating the number of distinct elements in the set intersection of two streams. We presented a complete analysis of the statistical performance (bias and variance) of three previously known schemes. We then computed the likelihood function of the problem and used it to present a new estimator, based on the ML method. We also found the optimal variance of any unbiased set intersection estimator, which is asymptotically achieved by our new ML scheme. 
We can conclude that our new scheme outperforms the three known schemes, significantly improves the variance (precision) of the estimator, and yields better results than the three previously known schemes.

\bibliographystyle{abbrv}
\bibliography{completePaper}


\appendix

\section*{Appendix}

\begin{lemma} \label{lem:covAandB}\ \\
The covariance of $\widehat{\abs{A}}$ and $\widehat{\abs{B}}$ satisfies $\Cov{\widehat{a}}{\widehat{b}} = \frac{n \cdot a \cdot b}{m \cdot u}$, where $m$ is the number of hash functions used for the estimation of $\widehat{a}$ and $\widehat{b}$.
\end{lemma}
\begin{proof} \ \\
Denote $x_A$ and $x_B$ as the maximal hash values for $A$ and $B$ respectively. 
\ifLongversion
We will compute the covariance $\Cov{x_A}{x_B}$ for one hash function, then generalize to all hash functions and use the Delta Method to derive $\Cov{\widehat{a}}{\widehat{b}}$.
\fi
Using our notation for $I_1,I_2$ and $I_3$ from Section \ref{sub:likelihood}, and applying the linearity of expectation, we obtain:
\begin{align}\label{eq:xAdotxB}
\Ex{x_A \cdot x_B} &= \Ex{(x_A\cdot x_B) \cdot (I_1+I_2+I_3)} \nonumber \\
&= \Ex{(x_A \cdot x_B) \cdot I_1} + \Ex{(x_A \cdot x_B) \cdot I_2} + \Ex{(x_A \cdot x_B) \cdot I_3} \text{.} 
\end{align}
We consider each term separately:

\begin{description}

\item[1) $x_A = x_B$:] From Eq. (\ref{eq:case1}), \hfill \smallskip

\begin{equation*}
\Ex{(x_A \cdot x_B) \cdot I_1} = \int_{x_A = x_B} {n \cdot x_A^{u+1} \,dx} = 
\ifLongversion
\\
= n \cdot \frac{1}{u+2}\cdot x_A^{u+2} \mid_0^1 = 
\\
\fi
\frac{n}{u+2} \text{.}
\end{equation*} 

\item[2) $x_A < x_B$:] From Eq. (\ref{eq:case2}), \smallskip

\begin{align*}
\Ex{(x_A \cdot x_B) \cdot I_2} &= \iint_{x_A > x_B} {\beta \cdot x_B^{\beta} \cdot a \cdot x_A^a \,dx_A\,dx_B} = \int_0^1 {\beta \cdot x_B^{\beta}\,dx_B} \cdot \int_0^{x_B} {a \cdot x_A^{a}\,dx_A} = \\
\ifLongversion
&= \int_0^1 {\beta \cdot x_B^{\beta} \cdot \frac{a}{a+1} \cdot x_A^{a+1} \mid_0^{x_B} \,dx_B} = \\
&= \int_0^1 {\beta \cdot x_B^{u+1} \cdot \frac{a}{a+1} \,dx_B} = \\
&= \frac{a}{a+1} \cdot \frac{\beta}{u+2} \cdot x_B^{u+2} \mid_0^1 = \\
\fi
&= \frac{a}{a+1} \cdot \frac{\beta}{u+2} \text{.}
\end{align*}

\item[3) $x_B > x_A$:] \hfill \smallskip

This case is symmetrical to the second case. We get that
\begin{align*}
\Ex{(x_A \cdot x_B) \cdot I_3} = \frac{b}{b+1} \cdot \frac{\alpha}{u+2} \text{.}
\end{align*}

\end{description}
Substituting the three terms into Eq. (\ref{eq:xAdotxB}) yields that
\begin{align} \label{eq:mul}
\Ex{x_A \cdot x_B} &= \frac{n}{u+2} + \frac{a}{a+1}\cdot\frac{\beta}{u+2} + \frac{b}{b+1}\cdot\frac{\alpha}{u+2} = \nonumber \\
&= \frac{ab\cdot (u+2) + n}{(u+2)(a+1)(b+1)}\text{.}
\end{align}
Using the covariance definition, we obtain:
\begin{align} \label{eq:covi}
\Cov{x_A}{x_B} &= \Ex{x_A \cdot x_B} - \Ex{x_A}\cdot\Ex{x_B} = \frac{ab\cdot (u+2) + n}{(u+2)(a+1)(b+1)} - \frac{a}{a+1}\cdot \frac{b}{b+1} = \nonumber \\
&= \frac{n}{(u+2)(a+1)(b+1)} \text{.}
\end{align}
The first equality is due to the covariance definition, and the second equality is due to Eq. (\ref{eq:mul}) and Lemma \ref{lemma:maxHash}(b). 

Eq. (\ref{eq:covi}) states the covariance for one hash function. We can generalize it for all hash functions. Let $x_A^i$ and $x_B^i$ be the maximal hash values for the $i$th hash function, for $A$ and $B$ respectively. Then
\begin{equation} \label{eq:covAll}
\Cov{x_A^i}{x_B^j} = \begin{cases} 
      0 & i \neq j \\
      \frac{n}{(u+2)(a+1)(b+1)} & i=j \:\text{.}
   \end{cases}
\end{equation}
We are now ready to compute $\Cov{\widehat{a}}{\widehat{b}}$. 
To this end, we first need to choose the cardinality estimator that we will use to estimate $\widehat{a}$ and $\widehat{b}$. For simplicity, we use the estimator from \cite{Cosma:2011}:
\begin{equation} \label{eq:est}
\widehat{a} = \frac{m}{\sum_{i=1}^{m}{(1-x_A^i)}} \text{,}
\end{equation}
where $m$ hash functions are used (symmetrically for $\widehat{b}$). 
Denoting $X = \sum{(1-x_A^i)}$, we can rewrite the estimator as $\widehat{a} = \frac{m}{X}$ ($Y$ is symmetrically defined for $\widehat{b}$). For each $i$, $x_A^i$ is the maximum of $a$ uniformly distributed variables, and thus
\begin{equation}\label{eq:expxi}
\Ex{x_A^i} = \frac{a}{a+1} \:\:\:\: \text{, and} \:\:\:\: \Var{x_A^i} = \frac{a}{(a+1)^2(a+2)} \text{.}
\end{equation}
Both equalities follow the beta distribution of $x_A^i$ (Lemma \ref{lemma:maxHash}(b)). From Eq. (\ref{eq:expxi}) it follows that
\begin{align}\label{eq:eX}
\Ex{X} = \Ex{\sum{(1-x_A^i)}} = m \cdot {(1-\Ex{x_A^i})} = m \cdot {\frac{1}{a+1}} = \frac{m}{a+1}\text{,}
\end{align}

and the variance
\begin{align}\label{eq:vX}
\Var{X} = \Var{\sum{(1-x_A^i)}} = m \cdot {\Var{x_A^i}} = \frac{m a}{(a+1)^2(a+2)}\text{.}
\end{align}
From Eqs. (\ref{eq:eX}) and (\ref{eq:vX}) we can conclude that $X \to \Normal{\frac{m}{a+1}}{\frac{m a}{(a+1)^2(a+2)}}$, and symmetrically for $Y$.
Recall that the estimators can be rewritten as $\widehat{a} = \frac{m}{X}$ (symmetrically for $\widehat{b}$). Applying the Delta Method (multivariate version) for the estimator's vector $(\widehat{a},\widehat{b})=(m/X, m/Y)$ yields
\begin{equation*}
\left( \begin{array}{c}
\widehat{a} \\
\widehat{b} \end{array} \right)
= 
\left( \begin{array}{c}
m/X \\
m/Y \end{array} \right)
\to
\Normal{\left( \begin{array}{c}
a \\
b \end{array} \right)}
{g' \cdot \Sigma \cdot g'^{T}} \text{,}
\end{equation*} 

where $g(X,Y) = (m/X,m/Y)$ is the function used in the Delta Method, $g'$ is its partial-derivatives matrix and $\Sigma$ is the variance matrix:
\begin{equation*}
\Sigma =
\left( \begin{array}{cc}
\Var{X} & \Cov{X}{Y} \\
\Cov{X}{Y} & \Var{Y} \end{array} \right)\text{.}
\end{equation*}
Using covariance properties and Eq. (\ref{eq:covAll}), we obtain:
\begin{equation}\label{eq:covXY}
\Cov{X}{Y} = \Cov{\sum{1-x_A^i}}{\sum{1-x_A^j}} = \sum_{i=1}^{m}{\Cov{x_A^i}{x_B^i}} = \frac{m n}{(u+2)(a+1)(b+1)}\text{.}
\end{equation}
Substituting the terms from Eqs. (\ref{eq:vX}) and (\ref{eq:covXY}) in $\Sigma$ yields
\begin{equation*}
\Sigma =
\left( \begin{array}{cc}
\frac{ma}{(a+1)^2(a+2)} & \frac{mn}{(a+1)(b+1)(u+2)} \\
\frac{mn}{(a+1)(b+1)(u+2)} & \frac{mb}{(b+1)^2(b+2)} \end{array} \right)\text{.}
\end{equation*}
Computing $g' \cdot \Sigma \cdot g'^{T}$ yields the final distribution of the estimators:
\begin{equation*}
\left( \begin{array}{c}
\widehat{a} \\
\widehat{b} \end{array} \right)
\to
\Normal{\left( \begin{array}{c}
a \\
b \end{array} \right)}
{\left( \begin{array}{cc}
\frac{a(a+1)^2}{m(a+2)} & \frac{n(a+1)(b+1)}{m(u+2)} \\
\frac{n(a+1)(b+1)}{m(u+2)} & \frac{b(b+1)^2}{m(b+2)} \end{array} \right)}\text{.}
\end{equation*} 
Finally, according to the Delta Method, $\Cov{\widehat{a}}{\widehat{b}}$ is the term in place $[1,2]$ in the matrix.
\end{proof}

\begin{lemma} \label{lem:covAandU}\ \\
The covariance of $\widehat{\abs{A}}$ and $\widehat{\abs{A \cup B}}$ satisfies $\Cov{\widehat{a}}{\widehat{u}} = \frac{a^2}{m} + \frac{n \cdot a \cdot b}{ m \cdot u} - \frac{u \cdot n \cdot a^2(b^2 + \alpha\beta)}{m \cdot Z}$, where $m$ is the number of storage units, and $Z$ satisfies:
\begin{equation*}
Z = {\alpha\cdot n}(a^2 + \alpha\beta) + \beta \cdot n(b^2 + \alpha\beta) + (a^2 + \alpha\beta)(b^2 + \alpha\beta)\text{.}
\end{equation*}
\end{lemma}
\begin{proof} \ \\
According to Fisher information matrix properties, $\widehat{\Cov{\widehat{a}}{\widehat{n}}}=(\mathbb{F}^{-1})_{1,3}$ is a Maximum Likelihood estimator for $\Cov{\widehat{a}}{\widehat{n}}$, 
where $(\mathbb{F}^{-1})_{1,3}$ is the term in place $[1,3]$ in the inverse Fisher information matrix. Computing the term $(\mathbb{F}^{-1})_{1,3}$ yields that
\begin{equation}\label{eq:covAandN}
\Cov{\widehat{a}}{\widehat{n}} = \frac{u \cdot n \cdot a^2(b^2 + \alpha\beta)}{m\cdot Z} \text{.}
\end{equation}
Therefore,
\begin{align*} \label{eq:varaplusb}
\Cov{\widehat{a}}{\widehat{u}} &= \Cov{\widehat{a}}{\widehat{a}+\widehat{b}-\widehat{n}} = \Cov{\widehat{a}}{\widehat{a}} + \Cov{\widehat{a}}{\widehat{b}} - \Cov{\widehat{a}}{\widehat{n}} = \\
&= \frac{a^2}{m} + \frac{n \cdot a \cdot b}{ m \cdot u} - \frac{u \cdot n \cdot a^2(b^2 + \alpha\beta)}{m\cdot Z} \text{.}
\end{align*}
The first equality is due to the inclusion-exclusion principle, and the second is due to covariance properties. The third equality is due to covariance properties, Lemma \ref{lemma:hyper}, Lemma \ref{lem:covAandB} and Eq. (\ref{eq:covAandN}).
\end{proof}
Similarly, we can obtain the covariance of $\widehat{\abs{B}}$ and $\widehat{\abs{A \cup B}}$:
\begin{equation*}
\Cov{\widehat{b}}{\widehat{u}} = \frac{b^2}{m} + \frac{n \cdot a \cdot b}{ m \cdot u} - \frac{u \cdot n \cdot b^2(a^2 + \alpha\beta)}{m \cdot Z} \text{.}
\end{equation*}

\end{document}